\documentclass[runningheads]{llncs}
\usepackage{amssymb,amsmath,amsthm}
\usepackage{mathtools}
\usepackage{graphicx}
\usepackage[T1]{fontenc}
\usepackage[linesnumbered,ruled,vlined,noend]{algorithm2e}
\usepackage{booktabs}
\usepackage{array}
\usepackage{colortbl}
\usepackage{units}
\usepackage{multirow,makecell}
\usepackage{pgfplots}
\usepackage{url}
\usepackage{float}
\usepackage{stmaryrd}
\usepackage{subcaption}

\usepackage{bookmark}

\usepackage{hyperref}

\hypersetup{
  colorlinks   = true, 
  urlcolor     = blue, 
  linkcolor    = blue, 
  citecolor   = blue 
}

\usepackage{etoolbox}
\usepackage[binary-units]{siunitx}
\robustify\bfseries
\usepackage{arydshln}

\usepackage{marvosym}

\renewcommand{\vec}[1]{\boldsymbol{#1}}

\newcommand{\sem}[1]{\llbracket {#1} \rrbracket}

\newcommand{\CL}{\ensuremath{\mathrm{CL}}}
\newcommand{\cl}{\ensuremath{\mathit{cl}}}
\newcommand{\bbI}{\ensuremath{\mathbb{I}}}
\newcommand{\bbP}{\ensuremath{\mathbb{P}}\xspace}
\newcommand{\bbF}{\ensuremath{\mathbb{F}}}

\newcommand{\Nat}{\ensuremath{\mathbb{N}}\xspace}
\newcommand{\Ints}{\mathbb{Z}}

\newcommand*{\defeq}{\stackrel{\text{def}}{=}}

\newif\ifmarkingSep

\newcommand{\markingScan}[2]{%
  \ifx\relax#1\empty
  \else
    \ifmarkingSep
      {\ }\relax
    \else
      \markingSeptrue
    \fi
    {#1}{\ast}{#2}\relax
    \expandafter\markingScan
  \fi
}

\makeatletter
\def \rightarrowfill{\m@th\mathord{\smash-}\mkern-6mu%
  \cleaders\hbox{$\mkern-2mu\mathord{\smash-}\mkern-2mu$}\hfill
  \mkern-6mu\mathord\rightarrow}
\makeatother

\makeatletter
\def \Rightarrowfill{\m@th\mathord{\smash-}\mkern-6mu%
  \cleaders\hbox{$\mkern-2mu\mathord{\smash-}\mkern-2mu$}\hfill
  \mkern-6mu\mathord\Rightarrow}
\makeatother

\makeatletter
\def \rightarrowfill{\m@th\mathord{\smash-}\mkern-6mu%
  \cleaders\hbox{$\mkern-2mu\mathord{\smash-}\mkern-2mu$}\hfill
  \mkern-6mu\mathord\rightarrow}
\makeatother

\makeatletter
\def \Rightarrowfill{\m@th\mathord{\smash=}\mkern-6mu%
  \cleaders\hbox{$\mkern-2mu\mathord{\smash=}\mkern-2mu$}\hfill
  \mkern-6mu\mathord\Rightarrow}
\makeatother

\makeatletter
\def \midrightarrowfill{\m@th\mathord{\smash{\raisebox{.2ex}{$\scriptscriptstyle\mid$}}\!\!\,-}\mkern-6mu%
  \cleaders\hbox{$\mkern-2mu\mathord{\smash-}\mkern-2mu$}\hfill
  \mkern-6mu\mathord\rightarrow}
\makeatother

\makeatletter
\def \midRightarrowfill{\m@th\mathord{\smash{\raisebox{.1ex}{$\scriptstyle\mid$}}\!\!\!=}\mkern-6mu%
  \cleaders\hbox{$\mkern-2mu\mathord{\smash=}\mkern-2mu$}\hfill
  \mkern-6mu\mathord\Rightarrow}
\makeatother

\newcommand{\overstackrel}[2]{\mathrel{\mathop{#1}\limits^{#2}}}
\newcommand{\sdefeq}{\mathbin{\smash[t]{\overstackrel{=}{\text{def}}}}}
\newcommand{\trans}[1]{\mathbin{\smash[t]{\overstackrel{\rightarrowfill}{\ #1\ }}}}
\newcommand{\wtrans}[1]{\mathbin{\smash[t]{\overstackrel{\Rightarrowfill}{\ #1\ }}}} 

\renewcommand{\orcidID}[1]{\smash{\href{http://orcid.org/#1}{\protect\raisebox{-1.25pt}{\protect\includegraphics{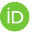}}}}}
\begin{document}
\title{Property Directed Reachability\\ for Generalized Petri Nets}
\titlerunning{PDR for Generalized Petri Nets}
\author{Nicolas Amat\inst{1}\orcidID{0000-0002-5969-7346}
  \and Silvano Dal Zilio\inst{1}\orcidID{0000-0002-6002-2696}
  \and Thomas Hujsa\inst{1}\orcidID{0000-0001-5226-8752}}
\authorrunning{N. Amat et al.}
\institute{LAAS-CNRS, Université de Toulouse, CNRS, INSA, Toulouse, France\\
\email{namat@laas.fr}}
\maketitle
\begin{abstract}
  We propose a semi-decision procedure for checking generalized
  reachability properties, on generalized Petri nets, that is based on
  the Property Directed Reachability (PDR) method. We actually define
  three different versions, that vary depending on the method used for
  abstracting possible witnesses, and that are able to handle problems
  of increasing difficulty. We have implemented our methods in a
  model-checker called \textsc{SMPT} and give empirical evidences that
  our approach can handle problems that are difficult or
  impossible to check with current state of the art tools.
\end{abstract}
%
%
%

\section{Introduction}
\label{sec:introduction}

We propose a new semi-decision procedure for checking reachability
properties on generalized Petri nets, meaning that we impose no
constraints on the weights of the arcs and do not require a finite
state space. We also consider a generalized notion of reachability, in
the sense that we can not only check the reachability of a given
state, but also if it is possible to reach a marking that satisfies a
combination of linear constraints between places, such as
$(p_0 + p_1 = p_2 + 2) \wedge (p_1 \leqslant p_2)$ for
example. Another interesting feature of our approach is that we are
able to return a ``certificate of invariance'', in the form of an
\emph{inductive linear invariant}~\cite{leroux2009general}, when we
find that a constraint is true on all the reachable markings. To the
best of our knowledge, there is no other tool able to compute such
certificates in the general case.

Our approach is based on an extension of the Property Directed
Reachability (PDR) method, originally developed for hardware
model-checking~\cite{jhala_sat-based_2011,hutchison_understanding_2012},
to the case of Petri nets. We actually define three variants of our
algorithm---two of them completely new when compared to our previous
work~\cite{pn2021}---that vary based on the method used for
generalizing possible witnesses and can handle problems of increasing
difficulty.

Reachability for Petri nets is an important and difficult problem with
many practical applications: obviously for the formal verification of
concurrent systems, but also for the study of diverse types of
protocols (such as biological or business processes); the verification
of software systems; the analysis of infinite state systems; etc.  It
is also a timely subject, as shown by recent publications on this
subject~\cite{blondin_directed_2021,dixon_kreach_2020}, but also with
the recent progress made on settling its theoretical
complexity~\cite{czerwinski2020reachability,DBLP:journals/corr/abs-2104-13866},
which asserts that reachability is Ackermann-complete, and therefore
inherently more complex than, say, the coverability problem.

A practical consequence of this ``inherent complexity'', and a general
consensus, is that we should not expect to find a one-size-fits-all
algorithm that could be usable in practice. A better strategy is to
try to improve the performances on some cases---for example by
developing new tools, or optimizations, that may perform better on
some examples---or try to improve ``expressiveness''---by finding
algorithms that can manage new cases, that no other tool can handle.

This wisdom is illustrated by the current state of the art at the
Model Checking Contest (MCC)~\cite{mcc2019}, a competition of
model-checkers for Petri nets that includes an examination for the
reachability problem. Albeit strongly oriented towards the analysis of
bounded nets. As a matter of fact, the top three tools in recent
competitions---\textsc{ITS-Tools}~\cite{its_tools},
\textsc{LoLA}~\cite{lola}, and \textsc{Tapaal}~\cite{tapaal}---all
rely on a portfolio approach.  Methods that have been proposed in this
context include the use of symbolic techniques, such as
$k$-induction~\cite{thierry-mieg_structural_2020}; abstraction
refinement~\cite{cassez2017refinement}; the use of standard
optimizations with Petri nets, like stubborn sets or structural
reductions; the use of the ``state equation''; reduction to integer
linear programming problems; etc.

The results obtained during the MCC highlight the very good
performances achieved when putting all these techniques together, on
bounded nets, with a collection of randomly generated properties.
Another interesting feedback from the MCC is that simulation
techniques are very good at finding a \emph{counter-example} when a
property is not an
invariant~\cite{blondin_directed_2021,thierry-mieg_structural_2020}.

In our work, we seek improvements in terms of both \emph{performance}
and \emph{expressiveness}. We also target what we consider to be a
difficult, and less studied area of research: procedures that can be
applied when a property is an invariant and when the net is unbounded,
or its state space cannot be fully explored. We also focus on the
verification of ``genuine'' reachability constraints, which are not
instances of a coverability problem. These properties are seldom
studied in the context of unbounded nets. Interestingly enough, our
work provides a simple explanation of why coverability problems are
also ``simpler'' in the case of PDR; what we will associate with the
notion of \emph{monotonic formulas}.

Concerning performances, we propose a method based on a well-tried
symbolic technique, PDR, that has proved successful with unbounded
model-checking and when used together with SMT
solvers~\cite{cimatti2014ic3,hoder2012generalized}. Concerning
expressiveness, we define a small benchmark of ``difficult nets'': a
set of synthetic examples, representative of patterns that can make
the reachability problem harder.

\subsubsection*{Outline and Contributions.}
We define background material on Petri nets in
Sect.~\ref{sec:petri-nets}, where we use Linear Integer Arithmetic
(LIA) formulas to reason about nets.
Section~\ref{sec:pdr} describes our decision method, based on PDR and
SMT solvers, for checking the satisfiability of linear invariants over
the reachable states of a Petri net.
Our method builds sequences of incremental invariants using both a
property that we want to disprove, and a stepwise approximation of the
reachability relation. It also relies on a generalization step where
we can abstract possible ``bad states'' into clauses
that are propagated in order to find a counter-example, or to block
inconsistent states.

We describe a first generalization method, based on the upset of
markings, that is able to deal with coverability properties. We
propose a new, dual variant based on the concept of
\emph{hurdles}~\cite{hack1976decidability}, that is without
restrictions on the properties. In this method, the goal is to block
bad sequences of transitions instead of bad states. We show how this
approach can be further improved by defining a notion of saturated
transition sequence, at the cost of adding universal quantification in
our SMT problems.

We have implemented our approach in an open-source tool, called
\textsc{SMPT}, and compare it with other existing tools. In
this context, one of our contributions is the definition of a set of
difficult nets, that characterizes classes of difficult reachability
problems.


\section{Petri Nets and Linear Reachability Constraints}
\label{sec:petri-nets}

Let $\Nat$ denote the set of natural numbers and $\Ints$ the set of
integers. Assuming $P$ is a finite, totally ordered set
$\{p_1, \dots, p_n\}$, we denote by $\Nat^P$ the set of mappings from
$P \to \Nat$ and we overload the addition, subtraction and comparison
operators ($=, \geq, \leq$) to act as their component-wise equivalent
on mappings. A QF-LIA formula $F$, with support in $P$, is a Boolean
combination of atomic propositions of the form $\alpha \sim \beta$,
where $\sim$ is one of $=, \leq$ or $\geq$ and $\alpha, \beta$ are
\emph{linear expressions}, that is, linear combinations of elements in
$\Nat \cup P$. We simply use the term \emph{linear constraint} to
describe $F$.

A \emph{Petri net} $N$ is a tuple
$(P, T, \textbf{pre}, \textbf{post})$ where $P = \{p_1, \dots, p_n\}$
is a finite set of places, $T$ is a finite set of transitions
(disjoint from $P$), and $\textbf{pre} : T \to \Nat^P$ and
$\textbf{post} : T \to \Nat^P$ are the pre- and post-condition
functions (also called the flow functions of $N$). A state $m$ of a
net, also called a \emph{marking}, is a mapping of $\Nat^P$. We say
that the marking $m$ assigns $m(p_i)$ {tokens} to place $p_i$. A
marked net $(N, m_0)$ is a pair composed from a net and an initial
marking $m_0$.

A transition $t \in T$ is \textit{enabled} at marking $m \in \Nat^P$
when $m \geqslant \textbf{pre}(t)$. When $t$ is enabled at $m$, we can
fire it and reach another marking $m' \in \Nat^P$ such that
$m' = m - \textbf{pre}(t) + \textbf{post}(t)$. We denote this
transition $m \trans{t} m'$. The difference between $m$ and $m'$ is a
mapping $\Delta(t) = \textbf{post}(t) - \textbf{pre}(t)$ in $\Ints^P$,
also called the \emph{displacement} of $t$.

By extension, we say that a \textit{firing sequence}
$\sigma = t_1\, \dots\, t_k \in T^*$ can be fired from $m$, denoted
$m \wtrans{\sigma} m'$, if there exist markings $m_0, \dots, m_k$ such
that $m = m_0$, $m' = m_k$ and $m_i \trans{t_{i+1}} m_{i+1}$ for all
$i < k$. We can also simply write $m \to^\star m'$. In this case, the
displacement of $\sigma$ is the mapping
$\Delta(\sigma) = \Delta({t_1}) + \dots + \Delta({t_k})$.
We denote by $R(N, m_0)$ the set of markings reachable from $m_0$ in
$N$. A marking $m$ is $k$-{bounded} when each place has at most $k$
tokens. By extension, we say that a marked Petri net $(N, m_0)$ is
bounded when there is $k$ such that all reachable markings are
$k$-bounded.

\begin{figure}[tb]
  \centering
  {
    \begin{subfigure}[t]{.48\textwidth}
      \centering
      \includegraphics[width=0.63\textwidth]{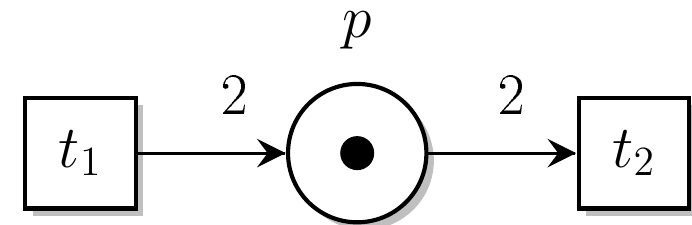}
    \end{subfigure}
    \begin{subfigure}[c]{.48\textwidth}
      \centering
      \includegraphics[width=\textwidth]{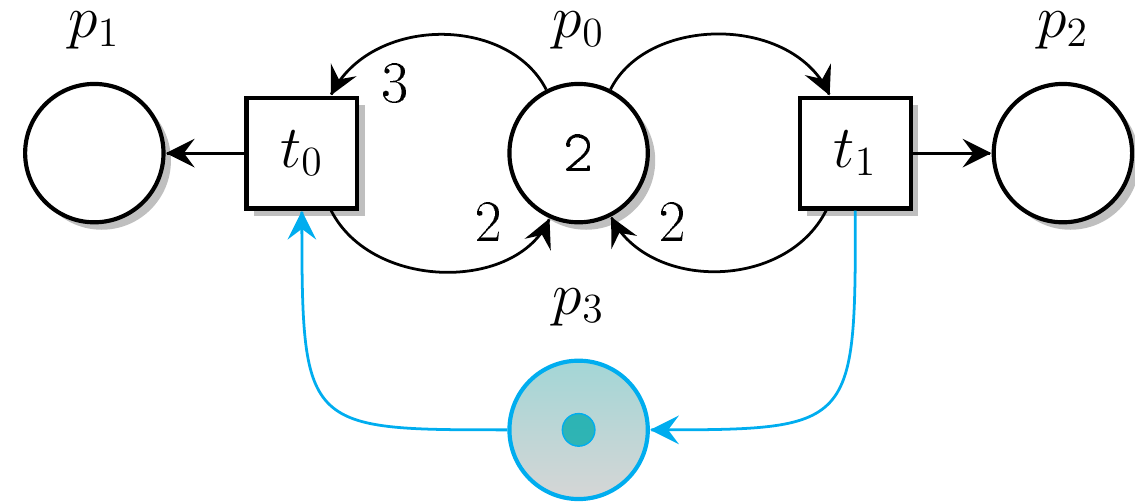}
    \end{subfigure}\hfill
  }
  \caption{Two examples of Petri nets: Parity (left) and PGCD (right).}
  \label{fig:nets}
\end{figure}

While reachable states are computed by adding a linear combination of
``displacements'' (vectors in $\Ints^P$), the set $R(N, m_0)$ is not
necessarily semilinear or, equivalently, definable using Presburger
arithmetic~\cite{ginsburg1966semigroups,leroux2009general}. This is a
consequence of the constraint that transitions must be enabled before
firing.  But there is still some structure to the set $R(N, m_0)$,
like for instance the following monotonicity constraint:
\begin{equation*}
  \forall m \in \Nat^P \,.\  m_1 \wtrans{\sigma} m_2 \ \text{ implies }\  m_1 +
  m \wtrans{\sigma} m_2 + m \tag{H1} \label{eq:h0}
\end{equation*}

We have other such results, such as with the notion of
\textit{hurdle}~\cite{hack1976decidability}. Just as $\textbf{pre}(t)$
is the smallest marking for which a given transition $t$ is enabled,
there is a smallest marking at which a given firing sequence $\sigma$
is fireable. This marking, denoted by $H(\sigma)$, has a simple
inductive definition:
\begin{equation*}
  H(t) = \textbf{pre}(t) \quad\text{ and }\quad
  H(\sigma_1 \cdot \sigma_2) = \max \left ( H(\sigma_1), H(\sigma_2) -
  \Delta(\sigma_1) \right ) \tag{H2} \label{eq:h1}
\end{equation*} 
Given this notion of hurdles, we obtain that $m \wtrans{\sigma} m'$ if
and only if (1) the sequence $\sigma$ is enabled: $m \geqslant H(\sigma)$,
and (2) $m'= m + \Delta(\sigma)$. We use this result in the second
variant of our method.

We can go a step further and characterize a necessary and sufficient
condition for firing the sequence $\sigma . \sigma^k$, meaning firing
the same sequence more than once.
Given $\Delta(\sigma)$, a place $p$ with a negative displacement (say
$-d$) means that we ``loose'' $d$ token each time we fire
$\sigma$. Hence we should budget $d$ tokens in $p$ for each new
iteration. On the opposite, nothing is needed for places with a
positive displacement, which accrue tokens.

Therefore we have $m \wtrans{\sigma} \wtrans{\sigma^k} m'$ if and only
if (1) $m \geqslant H(\sigma) + k \cdot \max(\vec{0}, - \Delta(\sigma))$,
and (2) $m'= m + (k+1) \cdot \Delta(\sigma)$.  Equivalently, if we
denote by $m^+$ the ``positive'' part of mapping $m$, such that
$m^+(p) = 0$ when $m(p) \leqslant 0$ and $m^+(p) = m(p)$ otherwise, we
have:
\begin{equation*}
  H(\sigma^{k+1}) = \max \left  ( H(\sigma),  H(\sigma) - k \cdot
    \Delta(\sigma) \right ) = H(\sigma) + k \cdot \left (
    - \Delta(\sigma) \right )^+  \tag{H3} \label{eq:h2}
\end{equation*}

\subsection{Examples}

We give two simple examples of unbounded nets in Fig.~\ref{fig:nets},
which are both part of our benchmark. Parity has a single place, hence
its state space can be interpreted as a subset of \Nat: with an
initial marking of $1$, this is exactly the set of odd numbers (and
therefore state $0$ is not reachable). We are in a special case where
the set $R(N, m_0)$ is semilinear. For instance, it can be seen as
solution to the constraint $\exists k . (p = 2 k + 1)$, or
equivalently $p \equiv 1 \pmod 2$. But it cannot be expressed with a
linear constraint involving only the variable $p$ without
quantification or modulo arithmetic. This example can be handled by
most of the tools used in our experiments, e.g. with the help of
$k$-induction.

In PGCD, transitions $t_0/t_1$ can decrement/increment the marking of
$p_0$ by $1$. Nonetheless, with this initial state, it is the case
that the number of occurrences of $t_0$ is always less than the one of
$t_1$ in any feasible sequence $\sigma$. Hence the two predicates
$p_0 \geq 2$ and $p_2 \geq p_1$ are valid invariants. (Since some
tools do not accept literals of the form $p \geq q$, we added the
``redundant'' place $p_3$ so we can restate our second invariant as
$p_3 \geq 1$.) These invariants cannot be proved by reasoning only on
the displacements of traces (using the state equation) and are already
out of reach for \textsc{LoLA} or \textsc{Tapaal}.

\subsection{Linear Reachability Formulas}

We can revisit the semantics of Petri nets using linear
predicates.
In the following, we use $\vec{p}$ for the vector $(p_1, \dots, p_n)$,
and $F(\vec{p})$ for a formula with variables in $P$. We also simply
use $F(\vec{\alpha})$ for the substitution
$F\{p_1 \leftarrow \alpha_1\}\dots\{p_n \leftarrow \alpha_n\}$, with
$\vec{\alpha} = (\alpha_1, \dots, \alpha_n)$ a sequence of linear
expressions.
We say that a mapping $m$ of $\Nat^P$ is a \emph{model} of $F$,
denoted $m \models F$, if the ground formula
$F(m) = F(m(p_1), \dots, m(p_n))$ is true. Hence we can also interpret
$F$ as a predicate over markings. Finally, we define the semantics of
$F$ as the set $\sem{F} = \{ m \in \Nat^P \mid m \models F\}$.

As usual, we say that a predicate $F$ is \emph{valid}, denoted
$\models F$, when all its interpretations are true
($\sem{F} = \Nat^P$); and that $F$ is \emph{unsatisfiable} (or simply
\texttt{unsat}), denoted $\nvDash F$, when $\sem{F} = \emptyset$.

We can define many properties on the markings of a net $N$ using this
For instance, we can model the set of markings $m$ such that some
transition $t$ is enabled using predicate $\mathrm{ENBL}_t$ (see
Equation~\eqref{eq:enbl} below).
We can also define a linear predicate to describe the relation between
the markings before and after some transition $t$ fires. To this end,
we use a vector $\vec{p'}$ of ``primed variables''
$(p'_1, \dots, p'_n)$, where $p'_i$ will stand for the marking of
place $p_i$ after a transition is fired.
With this convention, formula $\mathrm{FIRE}_t(\vec{p}, \vec{p'})$ is
such that $\mathrm{FIRE}_t(m, m')$ entails $m \trans{t} m'$ or
$m = m'$ when $t$ is enabled at $m$. With all these notations, we can
define a predicate $\mathrm{T}(\vec{p}, \vec{p'})$ that ``encodes''
the effect of firing at most one transition in the net $N$.

\begin{align}
  \mathrm{GEQ}_m(\vec{p}) &\ \defeq\  \textstyle \bigwedge_{i \in 1..n} \left ( p_i \geqslant
                            m(p_i) \right  )\\
  \mathrm{ENBL}_t(\vec{p}) &\ \defeq\  \textstyle \bigwedge_{i \in 1..n} \left ( p_i \geqslant
                             {\textbf{pre}(t)}(p_i) \right )  \ = \  \mathrm{GEQ}_{H(t)}(\vec{p})\label{eq:enbl}\\
  \Delta_t(\vec{p}, \vec{p'}) &\ \defeq\ \textstyle  \bigwedge_{i \in
                                1..n} \left ( p_i' =
                                p_i + \textbf{post}(t)(p_i) -
                                \textbf{pre}(t)(p_i) \right ) \label{eq:delta}\\
  \mathrm{EQ}(\vec{p}, \vec{p'}) &\ \defeq\ \textstyle \bigwedge_{i \in
                                   1..n} \left ( p'_i =
                                   p_i \right )\\
                               \displaybreak
  \mathrm{FIRE}_t(\vec{p}, \vec{p'}) &\ \defeq\ \mathrm{EQ}(\vec{p},
                                       \vec{p'}) \vee \left
                                       ({\mathrm{ENBL}}_t(\vec{p})
                                       \land \Delta_t(\vec{p},
                                       \vec{p'}) \right ) \\
  \mathrm{T}(\vec{p}, \vec{p'}) &\ \defeq\ \textstyle \mathrm{EQ}(\vec{p}, \vec{p'})
                                  \lor
                                  \bigvee_{t \in T} \left( \mathrm{ENBL}_t(\vec{p}) \land
                                  \Delta_t(\vec{p}, \vec{p'}) \right )
\end{align}

In our work, we focus on the verification of \textit{safety}
properties on the reachable markings of a marked net $(N,
m_0)$. Examples of properties that we want to check include: checking
if some transition $t$ is enabled (commonly known as
\emph{quasi-liveness}); checking if there is a deadlock; checking
whether some linear invariant between place markings is true; \dots
All properties that can be expressed using a linear predicate.

\begin{definition}[Linear Invariants and Inductive Predicates]
  A linear predicate $F$ is an invariant on $(N, m_0)$ if and only if
  we have $m \models F$ for all $m \in R(N,m_0)$. It is inductive if
  for all markings $m$ we have $m \models F$ and $m \to m'$ entails
  $m' \models F$.
\end{definition}

It is possible to characterize inductive predicates using our logical
framework. Indeed, $F$ is inductive if and only if the QF-LIA formula
(i) $F(\vec{p}) \land T(\vec{p}, \vec{p'}) \land \neg F(\vec{p'})$ is
\texttt{unsat}. Also, an inductive formula is an invariant when (ii)
$m_0 \models F$, or equivalently $\models F(m_0)$. As a consequence, a
sufficient condition for a predicate $F$ to be invariant is to have
both conditions (i) and (ii); conditions that can be checked using a
SMT solver. Unfortunately, the predicates that we need to check are
often not inductive. In this case, the next best thing is to try to
build an inductive invariant, say $R$, such that
$\sem{R} \subseteq \sem{F}$ (or equivalently $R \land \neg F$
\texttt{unsat}). This predicate provides a certificate of invariance
that can be checked independently.

\begin{lemma}[Certificate of Invariance]\label{lemma:certificate}
  A sufficient condition for $F$ to be invariant on $(N, m_0)$ is to
  exhibit a linear predicate $R$ that is (i) initial: $R({m_0})$
  \texttt{valid}; (ii) inductive:
  $R(\vec{p}) \land T(\vec{p}, \vec{p'}) \land \neg R(\vec{p\ })$
  \texttt{unsat}; and (iii) that entails $F$, for instance:
  $R \land \neg F$ \texttt{unsat}.
\end{lemma}

This result is in line with a property proved by
Leroux~\cite{leroux2009general}, which states that when a final
configuration $m$ is not reachable there must exist a Presburger
inductive invariant that contains $m_0$ but does not contain $m$. This
result does not explain how to effectively compute such an
invariant. Moreover, in our case, we provide a method that works with
general linear predicates, and not only with single configurations.
On the other side of the coin, given the known results about the
complexity of the problem, we do not expect our procedure to be
complete in the general case.

In the next section, we show how to (potentially) find such
certificates using an adaptation of the PDR method. An essential
component of PDR is to abstract a ``scenario'' leading to the model of
some property $F$---say a transition $m \wtrans{\sigma} m'$ with
$m' \models F$---into a predicate that contains $m$ (and potentially
many more similar scenarios). More generally, a \emph{generalization}
of the trio $(m, \sigma, F)$ is a predicate $G$ satisfied by $m$ such that
$m_1 \models G$ entails that there is $m_1 \to^\star m_2$ with
$m_2 \models F$.

We can use properties (H1)--(H3), defined earlier, to build
generalizations.

\begin{lemma}[Generalization\label{lemma:g}]
  Assume we have a scenario such that $m \wtrans{\sigma} m'$ and
  $m' \models F$. We have three possible generalizations of the trio
  $(m, \sigma, F)$.
  \begin{itemize}
  \item[\emph{(G1)}] If property $F$ is monotonic, then
    $m_1 \models \mathrm{GEQ}_m(\vec{p})$ implies there is
    $m_2 \geqslant m'$ such that $m_1 \wtrans{\sigma} m_2$ and
    $m_2 \models F$.
  \item[\emph{(G2)}] If
    $m_1 \models \mathrm{GEQ}_{H(\sigma)}(\vec{p}) \land F(\vec{p} +
    \Delta(\sigma))$ then $m_1 \wtrans{\sigma} m_2$ and
    $m_2 \models F$.
  \item[\emph{(G3)}] Assume $a, b$ are mappings of $\Nat^P$ such that
    ${a} = H(\sigma)$ and ${b} = \left ( - \Delta(\sigma) \right )^+$,
    with the notations used in (H3). Then
    \[
      m_1 \models \exists k . \left (
        \begin{array}[c]{l}
          \left [ \bigwedge_{i \in
          1..n} (p_i \geqslant a(i) + k \cdot b(i)) \right ]\\
          \ \land
          F(\vec{p} + (k + 1) \cdot \Delta(\sigma))\\
        \end{array}  \right )
      \ \text{ implies } \
      \left \{ \begin{array}[c]{l}
                 \exists k . m_1 \wtrans{\sigma^{k+1}} m_2\\
                 \text{ and } m_2 \models F
      \end{array} \right .
    \]
  \end{itemize}
\end{lemma}

\begin{proof}
  Each property is a direct result of properties (H1) to (H3).
\end{proof}

Property (G3) is the first and only instance of linear formula using
an extra variable, $k$, that is not in $P$. The result is still a
linear formula though, since we never need to use the product of two
variables. This generalization is used when we want to ``saturate the
sequence $\sigma$''. This is the only situation where we may need to
deal with quantified LIA formulas. Another solution would be to
replace each quantification with the use of modulo arithmetic, but
this operation may be costly and could greatly increase the size of
our formulas. It would also not cut down the complexity of the SMT
problems.

\section{Property Directed Reachability}
\label{sec:pdr}

Some symbolic model-checking procedure, such as
BMC~\cite{biere_symbolic_1999} or
$k$-induction~\cite{sheeran_checking_2000}, are a good fit when we try
to find counter-examples on infinite-state systems. Unfortunately,
they may perform poorly when we want to check an invariant. In this
case, adaptations of the PDR
method~\cite{jhala_sat-based_2011,hutchison_understanding_2012} (also
known as IC3, for ``Incremental Construction of Inductive Clauses for
Indubitable Correctness'') have proved successful.

We assume that we start with an initial state $m_0$ satisfying a
linear property, $\bbI$, and that we want to prove that property
$\bbP$ is an invariant of the marked net $(N, m_0)$. (We use
blackboard bold symbols to distinguish between parameters of the
problem, and formulas that we build for solving it.)
When checking for the reachability from the initial state, we can
simple choose $\bbI$ such that $\sem{\bbI} = \{ m_0\}$.

We define $\mathbb{F} = \neg \mathbb{P}$ as the ``set of feared
events''; such that $\mathbb{P}$ is not an invariant if we can find
$m$ in $R(N, m_0)$ such that $m \models \mathbb{F}$.
To simplify the presentation, we assume that $\bbF$ is a conjunction
of literals (a cube), meaning that $\bbP$ is a clause. In practice, we
assume that $\bbF$ is in Disjunctive Normal Form.

PDR is a combination of induction, over-approximation, and SAT or SMT
solving.
The goal is to build an incremental sequence of predicates
$F_0, \dots, F_k$ that are ``inductive relative to stepwise
approximations'': such that $m \models F_i$ and $m \to m'$ entails
$m' \models F_{i+1}$, but not $m' \models \mathbb{F}$.
The method stops when it finds a counter-example, or when we find that one of
the predicates $F_i$ is inductive.

We adapt the PDR approach to Petri nets, using linear predicates and
SMT solvers for the QF-LIA and LIA logics in order to learn,
generalize, and propagate new clauses. The most innovative part of our
approach is the use of specific ``generalization algorithms'' that
take advantage of the Petri nets theory, like the use of hurdles for
example.

\subsection{Algorithm}

Our implementation follows closely the algorithm for IC3 described
in~\cite{hutchison_understanding_2012}. We only give a brief sketch of
the OARS construction.

The main function, \hyperref[fun:prove]{\texttt{prove}}, computes an
\textit{Over Approximated Reachability Sequence} (OARS)
$(F_0, \dots, F_{k})$ of linear predicates, called \emph{frames}, with
variables in $\vec{p}$. An OARS meets the following constraints: (1)
it is monotonic: $F_i \land \neg F_{i+1}$ \texttt{unsat} for
$0 \leqslant i < k$; (2) it contains the initial states:
$\mathbb{I} \land \neg F_0$ \texttt{unsat}; (3) it does not contain
feared states: $F_{i} \land \mathbb{F}$ \texttt{unsat} for
$0 \leqslant i \leqslant k$; and (4) it satisfies \emph{consecution}:
${F_i}(\vec{p}) \wedge \mathrm{T}(\vec{p}, \vec{p'}) \land \neg
{F_{i+1}}(\vec{p'})$ \texttt{unsat} for $0 \leqslant i < k$.

By construction, each {frame} $F_i$ in the OARS is defined as a
set of clauses, $\CL({F_i})$, meaning that ${F_i}$ is built as a
formula in CNF: $F_i = \bigwedge_{\cl \in \CL(F_i)} \cl$. We also enforce
that $\CL(F_{i+1}) \subseteq \CL(F_i)$ for $0 \leqslant i < k$, which
means that the monotonicity property between frames is trivially
ensured.

The body of function \hyperref[fun:prove]{\texttt{prove}} contains a \emph{main iteration}
(line~\ref{l:main_iteration}) that increases the value of $k$ (the
number of levels of the OARS). At each step, we enter a second, minor
iteration (line~\ref{alg:minor_iteration} in function
\hyperref[fun:strengthen]{\texttt{strengthen}}), where we generate new minimal inductive clauses
that will be propagated to all the frames. Hence both the length of
the OARS, and the set of clauses in its frames, increase during
computation.

The procedure stops when we find an index $i$ such that
$F_i = F_{i+1}$. In this case we know that $F_i$ is an inductive
invariant satisfying $\bbP$. We can also stop during the iteration if
we find a counter-example (a model $m$ of $\bbF$). In this case, we
can also return a trace leading to $m$.

\begin{function}[tb]
  \DontPrintSemicolon 
  \SetKwFunction{sat}{sat}
  \SetKwFunction{strengthen}{strengthen}
  \SetKwFunction{propagateclauses}{propagateClauses}
  
  \KwResult{$\bot$ if $\bbF$ is reachable ($\bbP = \neg\bbF$ is not an invariant), otherwise $\top$} \BlankLine

  \If{\sat{$\bbI(\vec{p}) \land T(\vec{p}, \vec{p'}) \land \bbF(\vec{p'})$}} {\KwRet{$\bot$}} \BlankLine

  $k \gets 1$,  $F_0 \gets \bbI$, $F_1 \gets \bbP$\;
 
  \While{$\top$\label{l:main_iteration}}{\If{\textbf{not} \strengthen{$k$}}{\KwRet{$\bot$}}

    \propagateclauses{$k$}\;

    \If{$\CL(F_i) = \CL(F_{i+1})$ \textbf{for some} $1 \leqslant  i \leqslant
      k$}{\KwRet{$\top$}}

    $k \gets k + 1$\;}

  \caption{prove(\bbI, \bbF: \protect{linear predicates})}
  \label{fun:prove}
\end{function}

When we start the first minor iteration, we have $k = 1$, $F_0 = \bbI$
and $F_{1} = \bbP$. If we have
$F_k(\vec{p}) \wedge T(\vec{p}, \vec{p'}) \wedge \mathbb{F}(\vec{p})$
\texttt{unsat}, it means that \bbP is inductive, so we can stop and
return that \bbP is an invariant. Otherwise, we proceed with the
strengthen phase, where each model of
$F_k(\vec{p}) \wedge T(\vec{p}, \vec{p'}) \wedge \mathbb{F}(\vec{p})$
becomes a potential counter-example, or \emph{witness}, that we need
to ``block'' (line $3$--$5$ of function \hyperref[fun:strengthen]{\texttt{strengthen}}).

Instead of blocking only one witness, we first generalize it into a
predicate that abstracts similar dangerous states (see the call to
\texttt{generalizeWitness}). This is done by applying one of the
three generalization results in Lemma~\ref{lemma:g}. We give more
details about this step later. By construction, each generalization
is a cube $s$ (a conjunction of literals). Hence, when we block it, we
learn new clauses from $\neg s$ that can be propagated to the previous
frames.

\begin{function}[tb]
  \DontPrintSemicolon 

  \SetKwProg{try}{try}{:}{} \SetKwProg{catch}{catch}{:}{end}

  \SetKwFunction{sat}{sat}
  \SetKwFunction{generalizedwitness}{generalizeWitness}
  \SetKwFunction{inductivelygeneralize}{inductivelyGeneralize}
  \SetKwFunction{pushgeneralization}{pushGeneralization}

  \try{}
  {\While{$(m \trans{t} m') \models\ F_k(\vec{p}) \land T(\vec{p}, \vec{p'}) \land \bbF(\vec{p'})$\label{alg:minor_iteration}}
    {
      \textcolor{red}{$s \gets \generalizedwitness{m, t, \bbF}$}\;\label{l:strengten_witness}
      $n \gets \inductivelygeneralize{s, k - 2, k}$\;
    $\pushgeneralization{\{(s, n+1)\}, k}$\;
    }
    \KwRet{$\top$}}
    
    \BlankLine

  \catch{\text{counter example}}{\KwRet{$\bot$}}

  \caption{strengthen($k$ : \protect{current level})}
  \label{fun:strengthen}
\end{function}

\begin{procedure}[tb]
  \DontPrintSemicolon 

  \SetKwFunction{unsat}{unsat}
  
  \For{$i \gets 1$ \textbf{to} $k$} {\ForEach{$cl \in CL(F_i)$}{\If{{$\nvDash
  F_i(\vec{p}) \land T(\vec{p}, \vec{p'}) \land \neg cl(\vec{p'})$}}
  {$\CL(F_{i+1}) \gets \CL(F_i) \cup \{cl\}$\;}}} \caption{propagateClauses($k$:
  \protect{level})}
\label{proc:propagate_clauses}
\end{procedure}

\begin{procedure}[tb]
  \DontPrintSemicolon 

  \SetKwFunction{unsatcore}{unsat\_core}

  $cl \gets \neg \, $\unsatcore{$\neg s(\vec{p}) \land F_i(\vec{p}) \land
  T(\vec{p}, \vec{p'}) \land s(\vec{p'})$} \; \For{$j \gets 1$ \textbf{to} i+1}
  {$\CL(F_j) \gets \CL(F_j) \cup \{cl\}$}

  \caption{generateClause($s$ : cube,  $i$: level, $k$: \protect{level})}
  \label{proc:generate_clause}
\end{procedure}

Before pushing a new clause, we test whether $s$ is reachable from
previous frames. We take advantage of this opportunity to find if we
have a counter-example and, if not, to learn new clauses in the
process. This is the role of functions \hyperref[proc:push_generalization]{\texttt{pushGeneralization}}
and \hyperref[fun:inductively_generalize]{\texttt{inductivelyGeneralize}}.


We find a counter example (in the call to \hyperref[fun:inductively_generalize]{\texttt{inductivelyGeneralize}}) if
the generalization from a witness found at level $k$, say $s$, reaches level $0$
and $F_0(\vec{p}) \land T(\vec{p}, \vec{p'}) \land s(\vec{p'})$ is satisfiable
(line $1$ in \hyperref[fun:inductively_generalize]{\texttt{inductivelyGeneralize}}). Indeed, it means that we can
build a trace from $\bbI$ to $\bbF$ by going through $F_1, \dots, F_k$.

The method relies heavily on checking the satisfiability of linear
formulas in QF-LIA, which is achieved with a call to a SMT solver. In
each function call, we need to test if predicates of the form
$F_i \land T \land G$ are \texttt{unsat} and, if not, enumerate its
models. To accelerate the strengthening of frames, we also rely on the
unsat core of properties in order to compute a \emph{minimal inductive
  clause} (MIC).

\begin{function}[tb]
  \DontPrintSemicolon 
  
  \SetKwFunction{sat}{sat} \SetKwFunction{generateclause}{generateClause}

  \If{$min < 0$ \textbf{and} \sat{$F_0(\vec{p}) \land T(\vec{p}, \vec{p'}) \land s(\vec{p'})$}} {\textbf{raise} $Counterexample$}
  \BlankLine

  \For{$i \gets \max(1, min + 1)$ \textbf{to} $k$} {\If{\sat{$F_i(\vec{p}) \land
  T(\vec{p}, \vec{p'}) \land \neg s(\vec{p}) \land s(\vec{p'})$}}
  {$\generateclause{s, i-1, k}$\;\KwRet{$i - 1$}}} $\generateclause{s, k, k}$\;
  \KwRet{$k$}

  \caption{inductivelyGeneralize($s$ : cube,  $min$: level, $k$: \protect{level})}
  \label{fun:inductively_generalize}
\end{function}

\begin{function}[tb]
  \DontPrintSemicolon 

  \SetKwFunction{sat}{sat}
  \SetKwFunction{generalizedwitness}{generalizeWitness}
  \SetKwFunction{inductivelygeneralize}{inductivelyGeneralize}
  
  \While{$\top$} {$(s,n) \gets \text{from } states \text{ minimizing } n$\;

    \If{$n > k$} {\KwRet{}}

    \If{$(m \trans{t} m') \models F_n(\vec{p}) \land T(\vec{p}, \vec{p'}) \land s(\vec{p'})$} {\textcolor{red}{$p
    \gets \generalizedwitness{m, t, s}$}\;\label{l:pushgeneralization_witness}
     $l \gets \inductivelygeneralize{p, n - 2,
    k}$\; $states \gets states \cup \{(p, l + 1)\}$\;} \Else{$l \gets
    \inductivelygeneralize{s, n, k}$\; $states \gets states \setminus \{(s, n)\}
    \cup \{(s, l + 1)\}$\;}}

  \caption{pushGeneralization($states$: \text{set of} (state, level), $k$: \protect{level})}
  \label{proc:push_generalization}
\end{function}

Our approach is parametrized by a generalization function
(\texttt{generalizeWit\-ness}) that is crucial if we want to avoid
enumerating a large, potentially unbounded, set of witnesses. This can
be the case, for example, in line $5$
of~\hyperref[proc:push_generalization]{\texttt{pushGene\-ralization}}. In this particular case, we find a
state $m$ at level $n$ (because $m \models F_n$), and a transition $t$
that leads to a problematic clause in $F_{n+1}$. Therefore we have a
sequence $\sigma$ of size $k - n + 1$ such that $m \wtrans{\sigma} m'$
and $m' \models \bbF$. We consider three possible methods for
generalizing the trio $(m, \sigma, \bbF)$, that corresponds to
property (G1)--(G3) in Lemma~\ref{lemma:g}.


\subsection{State-based Generalization}

A special case of the reachability problem is when the predicate
$\bbF$ is monotonic,, meaning that $m_1 \models \bbF$ entails
$m_1 + m_2 \models \bbF$ for all markings $m_1, m_2$. A sufficient
(syntactic) condition is for $\bbF$ to be a positive formula with
literals of the form $\sum_{i \in I} p_i \geq a$.  This class of
predicates coincide with what is called a \emph{coverability
  property}, for which there exists specialized verification methods
(see e.g.~\cite{finkel1991minimal,DBLP:journals/deds/FinkelHK21}).

By property (G1), If we have to block a witness $m$ such that $m \wtrans{\sigma}
m'$ and $m' \models \bbF$, we can as well block all the states greater than $m$.
Hence we can choose the predicate $\mathrm{GEQ}_m$ to generalize $m$.
This is a very convenient case for verification and one of the
optimizations used in previous works on PDR for Petri
nets~\cite{pn2021,esparza_smt-based_2014,DBLP:conf/apn/KangBJ21,kloos_incremental_2013}.
First, the generalization is very simple and we can easily compute a
MIC when we block predicate $\mathrm{GEQ}_m$ in a frame. Also, we can
prove the completeness of the procedure when $\bbF$ is monotonic. An
intuition is that it is enough, in this case, to check the property on
the minimal coverability set of the net, which is always
finite~\cite{finkel1991minimal}. The procedure is also complete for
finite transition systems. These are the only cases where we have been
able to prove that our method always terminates.

\subsection{Transition-based Generalization}

We propose a new generalization based on the notion of hurdles. This
approach can be used when $\bbF$ is not monotonic, for example when we
want to check an invariant that contains literals of the form $p = k$
(e.g. the reachability of a fixed marking) or $p \geqslant q$.

Assume we need to block a witness of the from
$m \wtrans{\sigma} m' \models s$. Typically, $s$ is a cube in
$\bbF$, or a state resulting from a call to
\hyperref[proc:push_generalization]{\texttt{pushGeneralization}}. By property (G2), we can as well block
all the states satisfying
$G_\sigma(\vec{p}) \sdefeq \mathrm{GEQ}_{H(\sigma)}(\vec{p}) \land
s(\vec{p} + \Delta(\sigma))$. 
This generalization is interesting when property $s$ does not
constraint all the places, or when we have few equality
constraints. In this case $G_\sigma$ may have an infinite number of
models.
It should be noted that using the duality between ``feasible traces''
and hurdles is not new. For example, it was used
recently~\cite{DBLP:journals/deds/FinkelHK21} to accelerate the
computation of coverability trees. Nonetheless, to the best of our
knowledge, this is the first time that this generalization method has
been used with PDR.

\subsection{Saturated Transition-based Generalization}

We still assume that we start from a witness
$m \wtrans{\sigma} m' \models s$. Our last method relies on property
(G3) and allows us to consider several iterations of $\sigma$. If we
fix the value of $k$, then a possible generalization is
$G^k_\sigma \sdefeq \left ( \bigwedge_{i \in 1..n} (p_i \geqslant a(i)
  + k \cdot b(i)) \right ) \land s(\vec{p} + (k + 1) \cdot
\Delta(\sigma))$, where $a, b$ are the mappings of $\Nat^P$ defined in
Lemma~\ref{lemma:g}. (Notice that $G^1_\sigma = G_\sigma$.)
More generally the predicate
$G^{\leqslant k}_\sigma = G^1_\sigma \vee \dots \vee G^k_\sigma$ is a valid
generalization for the witness $(m, \sigma, s)$, in the sense that if
$m_1 \models G^{\leqslant k}_\sigma$ then there is a trace
$m_1 \to^\star m_2$ such that $m_2 \models s$.
At the cost of using existential quantification (and therefore a
``top-level'' universal quantification when we negate the predicate to
block it in a frame), we can use the more general predicate
$G^\star_\sigma \sdefeq \exists k . G^k_\sigma$, which is still linear
and has its support in $P$.

We know examples of invariants where the PDR method does not terminate
except when using saturation. A simple example is the net {Parity},
used as an example in Sect.~\ref{sec:petri-nets}, with the invariant
$\bbP = (p \geqslant 1)$. In this case, $\bbF = \neg \bbP = (p =
0)$. Hence we are looking for witnesses such that $m \to^\star 0$. The
simplest example is $2 \trans{t_2} 0$, which corresponds to the
``blocking clause'' $p \neq 2$.  In this case, we have $H(t_2) = 2$
and $\Delta(t_2) = - 2$. Hence the transition-based generalization is
$(p \geq 2) \land (p - 2 = 0) \equiv (p = 2)$, which does not block
new markings. At this point, we try to block $(p = 0) \lor (p =
2)$. The following minor iteration of our method will consider the
witness $4 \wtrans{t_2.t_2} 0$, etc. Hence after $k$ minor iterations,
we have $F_k \equiv (p \neq 0) \land (p \neq 2) \land \dots \land (p \neq 2 k)$.
If we saturate $t_2$, we find in one step that we should block
$\exists k . (p - 2 \cdot (k + 1) = 0)$. This is enough to prove that
$(p \geqslant 1)$ is an invariant as soon as the initial marking is an
odd number.

This example proves that PDR is not complete, without saturation, in
the general case. We conjecture that it is also the case with
saturation. Even though example Parity is extremely simple, it is also
enough to demonstrate the limit of our method without
saturation. Indeed, when we only allow unquantified linear predicates
with variables in $P$, it is not possible to express all the possible
semilinear sets in $\Nat^P$. (We typically miss some periodic sets.)
In practice, it is not always useful to saturate a trace and, in our
implementation, we use heuristics to limit the number of
quantifications introduced by this operation. Actually, nothing
prevents us from mixing our different kinds of generalization
together, and there is still much work to be done in order to find
good tactics in this case.


\section{Experimental Results}
\label{sec:experimental-results}

We have implemented our complete approach in a tool, called
\textsc{SMPT} (for Satisfiability Modulo P/T Nets), and made our code
freely available under the GPLv3 license. The software, scripts and
data used to perform our analyses are available on Github
(\url{htttps://github.com/nicolasAmat/SMPT}) and are archived in
Zenodo~\cite{zenodoamat}. The tool supports the declaration of
reachability constraints expressed using the same syntax as in the
Reachability examinations of the Model Checking Contest (MCC). For
instance, we use PNML as the input format for nets. \textsc{SMPT}
relies on a SMT solver to answer \texttt{sat} and \texttt{unsat-core}
queries. It interacts with SMT solvers using the SMT-LIBv2 format,
which is a well-supported interchange format. We used the \textsc{z3}
solver for all the results presented in this section.

\subsection{Evaluation on Expressiveness}

It is difficult to find
benchmarks with unbounded Petri nets. To quote Blondin et
al.~\cite{blondin_directed_2021}, ``due to the lack of tools handling
reachability for unbounded state spaces, benchmarks arising in the
literature are primarily coverability instances''. It is also very
difficult to randomly generate a true invariant that does not follow,
in an obvious way, from the state equation. For this reason, we
decided to propose our own benchmark, made of five synthetic examples
of nets, each with a given invariant. This benchmark is freely
available and presented as an archive similar to instances of problems
used in the MCC.

Our benchmark is made of deceptively simple nets that have been
engineered to be difficult or impossible to check with current
techniques. We already depicted our two first examples in
Fig.~\ref{fig:nets}. We display all our other examples in Figs.~$2$,
$3$ and~$4$.


\begin{figure}
  \centering
  \includegraphics[width=0.77\textwidth]{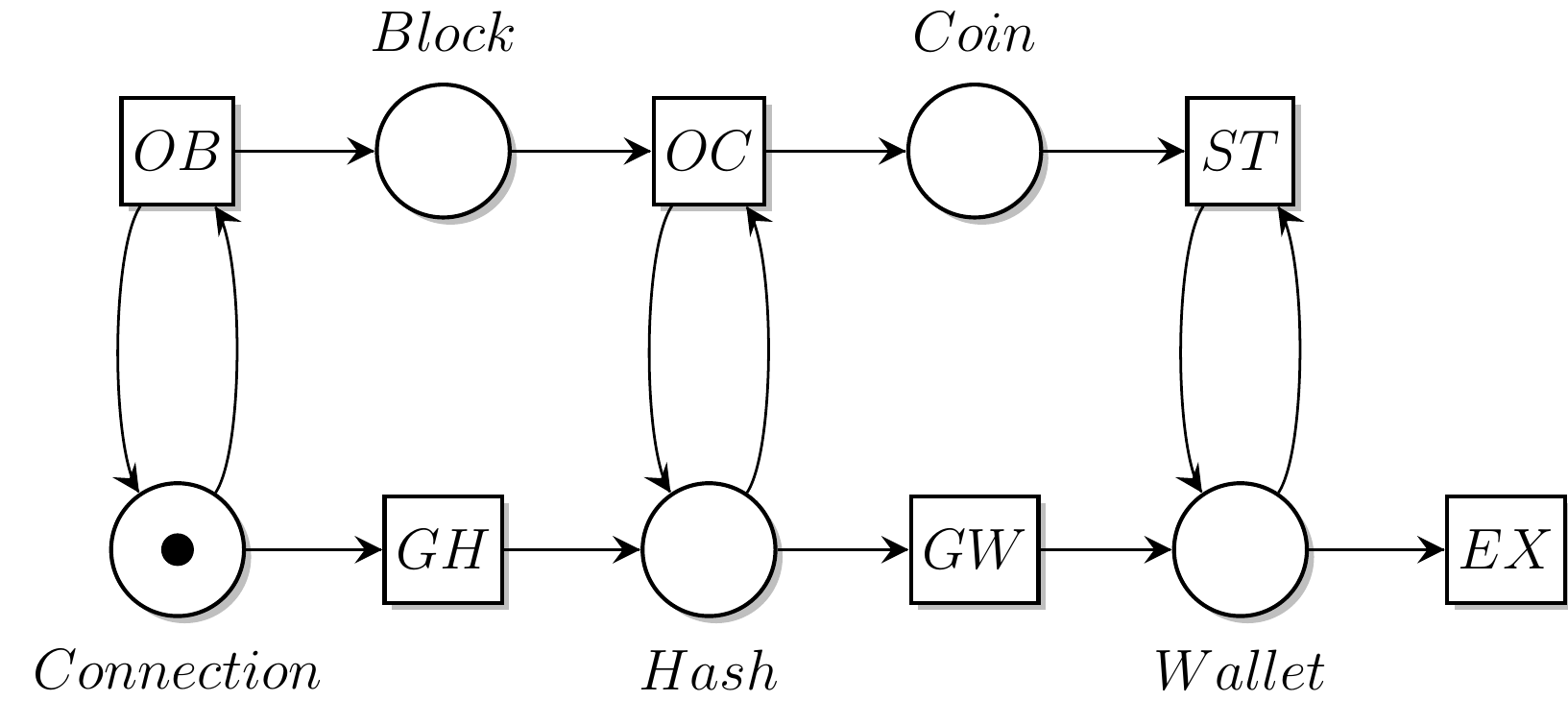}
  \caption{CryptoMiner with $\bbP = \neg (Block = 4 \land Connection = 1 \land Coin = 10)$}
  \label{fig:cryptominer}
\end{figure}

\begin{figure}
  \centering
  \includegraphics[width=0.55\textwidth]{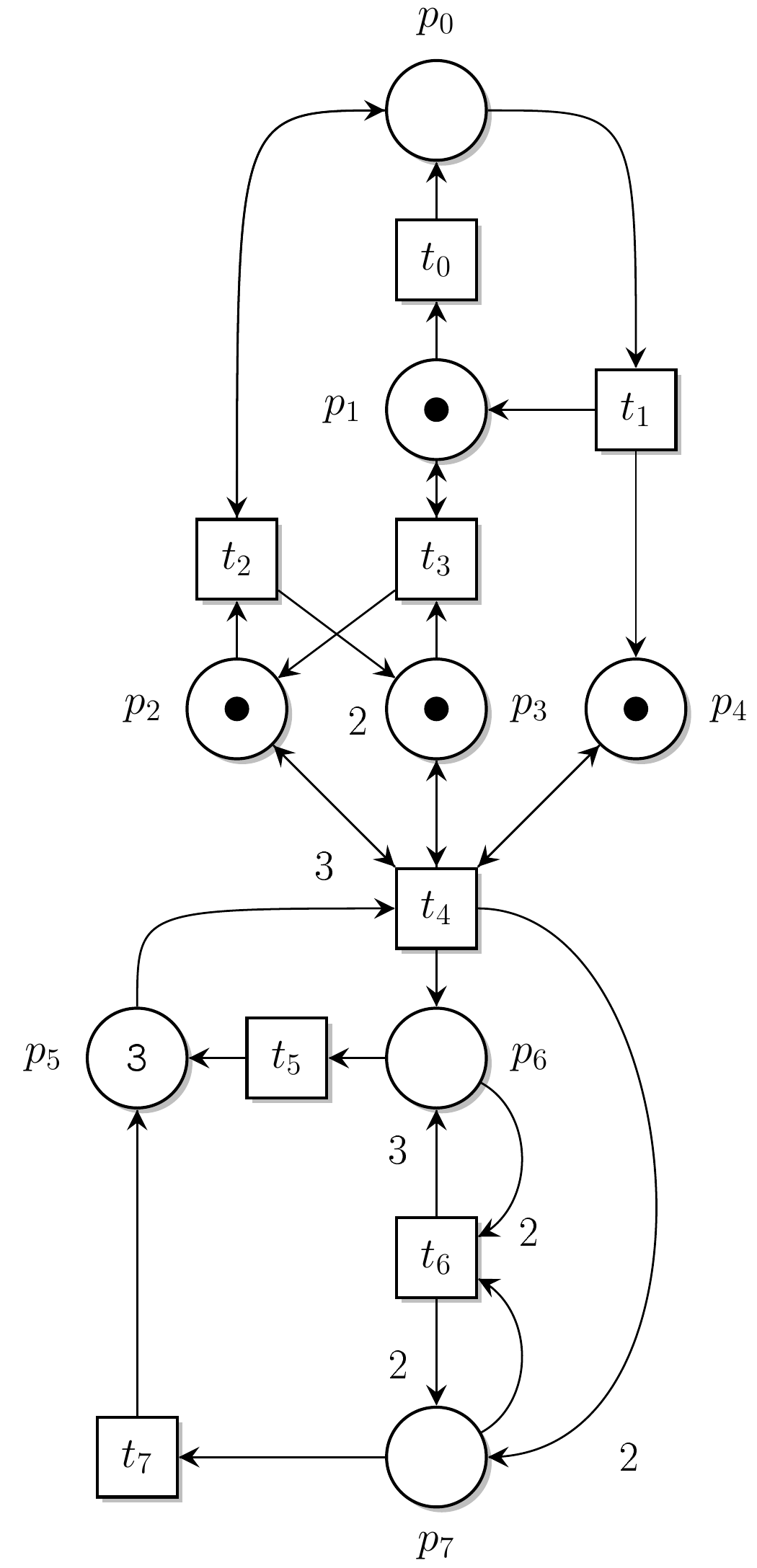}
  \caption{Process with $\bbP = (p_2 + p_3 + p_4 \geqslant 1 \land p_7 \leqslant 2)$}
  \label{fig:process}
\end{figure}

\begin{figure}
  \centering
  \includegraphics[width=0.88\textwidth]{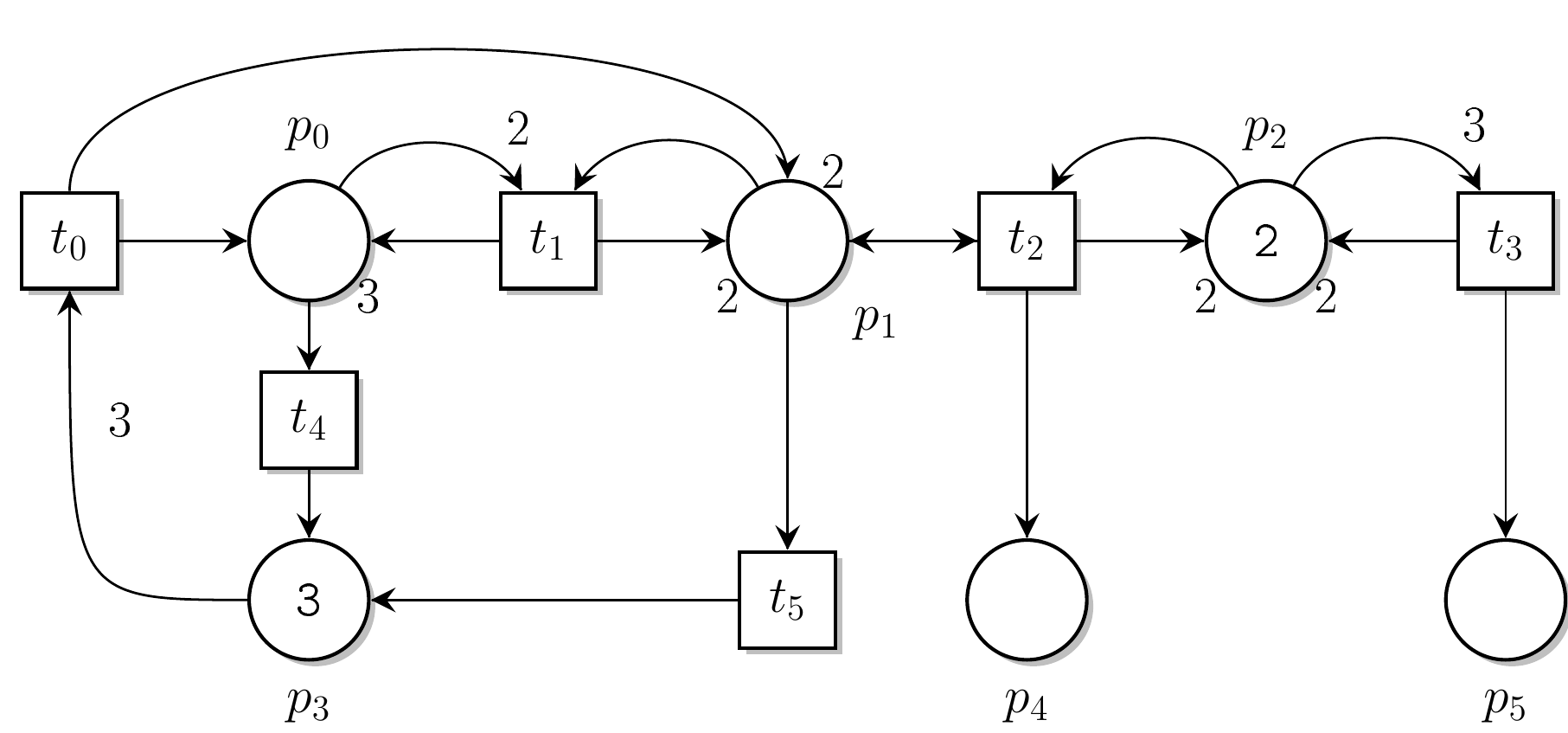}
  \caption{Murphy with $\bbP = (p_1 \leqslant 2 \land p_4 \geqslant p_5)$}
  \label{fig:murphy}
\end{figure}


\newcommand{\TLE}{\cellcolor{gray!8}\color{purple}{\makebox[5em]{TLE}}}
\newcommand{\WIN}[1]{\bfseries ${{\SI{#1}{}}}$}

\begin{table}[tb]
  \centering
\newcolumntype{x}[1]{>{\centering\arraybackslash\hspace{0pt}}p{#1}}
\begin{center}
  \begin {tabular}{l c@{\qquad} c@{}l l@{\quad} w{c}{5em} l@{\quad} w{c}{5em} l@{\quad} w{c}{5em}}%
    \toprule

    Instance && \multicolumn{2}{c}{\textsc{SMPT}}
    && \textsc{ITS-Tools} && \textsc{LoLA}
    && \textsc{Tapaal}\\\midrule
    Murphy && \WIN{0.75} &${}^*$ && \TLE && \TLE && \TLE\\

    PGCD && \WIN{0.11} &${}^*$ && \SI{139.08}{} && \TLE && \TLE\\

    CryptoMiner && \SI{0.19}{} &${}^*$ && \SI{5.92}{} && \TLE && \WIN{0.18}\\

    Parity  && \SI{0.40}{} &${}^*$ && \SI{3.36}{} && \WIN{0.01} && \SI{4.16}{}\\

    Process && \SI{83.39}{} & && \TLE && \WIN{0.03} && \SI{0.18}{}\\

    \bottomrule
  \end{tabular}
\end{center}
\caption{Computation time on our synthetic examples (time in seconds).\label{tab1}}
\end{table}

We give a brief description of the nets composing our ``reachability''
benchmark (except for Parity and PGCD from Fig.~\ref{fig:nets} already
described previously). Each of our example is quite small, with less
than $10$ places or transitions, and is representative of patterns
that can make the reachability problem harder: the use of self-loops;
dead transitions that cannot be detected with the state equation;
weights that are relatively prime; etc. Also, most of our examples can
be turned into families of nets using parameters such has the initial
marking, weights on the arcs, or by adding copies of a sub-net.

\begin{itemize}
\item \textbf{CryptoMiner} describes the, simplified, daily schedule
  of someone mining bitcoins. The net is composed of two disjoint
  state machines synchronized by self-loops (trivial cycles of weight
  1). Removing the self-loops do not modify the incidence matrix, and
  so do not change the solutions of the state equation. The difficulty
  when analysing this net lies in the presence of constraints that
  cannot be derived from the state equation alone. For instance, the
  presence of tokens in $\mathit{Coin}$ implies $\mathit{Connection}$
  empty.

\item \textbf{Process} is a net composed of three subnets coupled by
  self-loops on the places $p_2$, $p_3$ and $p_4$. The component at
  the bottom includes a dead transition ($t_6$); it will never be
  enabled although the state equation ensures at least one possibility
  of firing it. Like with our previous example, reasoning only on the
  state equation is not enough to capture the exact behaviour of this
  net. For instance, the state equation allows to get $3$ tokens in
  $p_7$, which would contradict our invariant.

\item \textbf{Murphy} is a net combining PGCD with the ``bottom
  component'' of net Process.
\end{itemize}

We compared \textsc{SMPT} against \textsc{ITS-Tools}, \textsc{LoLA},
and \textsc{Tapaal} and give our results in Table~\ref{tab1}. All
results are computed using $4$ cores, a limit of $\SI{16}{\giga\byte}$
of RAM, and a timeout of $\SI{1}{\hour}$. A result of TLE stands for
``Time Limit Exceeded''. For \textsc{SMPT}, we marked with an asterisk
(${}^*$) the results computed using our saturation-based
generalization.
Our results show that \textsc{SMPT} is able to answer on several
classes of examples that are out of reach for some, or all the other
tools; often by orders of magnitude.

We also experimented with two other tools for reachability recently
presented at TACAS: \textsc{KReach}~\cite{dixon_kreach_2020}, that
provides a complete implementation of Kosaraju's original decision
procedure, and \textsc{FastForward}~\cite{blondin_directed_2021}, a
tool for efficiently finding counter-examples in unbounded Petri nets
(but that may report that an invariant is true in some cases). We do
not include these tools in our findings since they were unable to
answer any of our problems. (But input files for our benchmark, for
both tools, are available in our artifact.)

\subsection{Computing Certificate of Invariance}

A distinctive
feature of \textsc{SMPT} is the ability to output a linear inductive
invariant for reachability problems: when we find that \bbP is
invariant, we are also able to output an inductive formula
$\mathbb{C}$, of the form $\bbP \land G$, that can be checked
independently with a SMT solver. We can find the same capability in
the tool \textsc{Petrinizer}~\cite{esparza_smt-based_2014} in the case
of coverability properties.

To get a better sense of this feature, we give the actual outputs
computed with \textsc{SMPT} on the two nets of
Fig.~\ref{fig:nets}. The invariant for the net Parity is
$\bbP_1 = (p_0 \geqslant 1)$, and for PGCD it is
$\bbP_2 = (p_1 \leqslant p_2)$

The certificate for property $\bbP_1$ on Parity is
$\mathbb{C}_1 \equiv (p_0 \geqslant 1) \land \forall k . ( (p_0 < 2\,k
+ 2) \lor (p_0 \geqslant 2\,k + 3) )$, which is equivalent to
$(p_0 \geqslant 1) \land (\forall k \geqslant 1) . (p_0 \neq 2.k)$,
meaning the marking of $p_0$ is odd. This invariant would be different
if we changed the initial
marking to an even number.\\[-1.5em]
{\footnotesize{%
\begin{verbatim}
[PDR] Certificate of invariance
# (not (p0 < 1))
# (forall (k1) ((p0 < (2 + (k1 * 2))) or (p0 + (-2 * (k1 + 1))) >= 1))
\end{verbatim}%
}}

The certificate for property $\bbP_2$ on PGCD is
$\mathbb{C}_2 \equiv (p_1 \leqslant p_2) \land \forall k . ((p_0 < k +
3) \lor (p_2 - p_1 \geqslant k + 1) )$ and may seem quite
inscrutable. It happens actually that the saturation ``learned'' the
invariant $p_0 + p_1 = p_2 + 2$ and was able to use this information
to strengthen property $\bbP_2$ into an inductive invariant.\\[-1.5em]
{\footnotesize{%
\begin{verbatim}
[PDR] Certificate of invariance
# (not (p1 > p2))
# (forall (k1) ((p0 < (3 + (k1 * 1))) or ((p1 + (1 * (k1 + 1))) <= p2)) 
\end{verbatim}%
}}

\subsection{Evaluation on Performance}

Since it is not sufficient
to use only a small number of hand-picked examples to check the
performance of a tool, we also provide results obtained on a set of
$30$ problems (a net together with an invariant) that are borrowed
from test cases used by the tool
\textsc{Sara}~\cite{saratata,wimmel2012applying}
 (examples test$\{\mathrm{3,4,12}\}$)
and a similar software, called
\textsc{Reach}, that is part of the \textsc{Tina}
toolbox~\cite{berthomieu2004tool}
 (examples 1, 3u, \dots, zz).
Most of these problems can be easily answered, but are interesting to
test our reliability on a relatively even-handed benchmark.

Our benchmark also include $6$ examples of bounded nets obtained by
limiting the number of times we can fire transitions in the nets PGCD
and CryptoMiner. (This is achieved by adding a new place that loses a
token when a transition is fired.)

The experiments were performed with the same conditions as previously,
but with a timeout of only 255s.  We display our results in the chart
of Fig.~\ref{fig:cactus}, which gives the number of feasible problems,
for each tool, when we change the timeout value. We also provide the
computation times, for the same dataset, in Table~\ref{tab2}. We
observe that our performances are on par with \textsc{Tapaal}, which
is the fastest among our three reference tools on this benchmark.

\renewcommand{\TLE}{\cellcolor{gray!8}\color{purple}{\makebox[4em]{TLE}}}

\begin{table}[tb]
  \centering
\newcolumntype{x}[1]{>{\centering\arraybackslash\hspace{0pt}}p{#1}}
\begin{center}
  \begin {tabular}{l c@{\qquad} c@{}l l@{\quad} w{c}{5em} l@{\quad} w{c}{5em} l@{\quad} w{c}{5em}}%
    \toprule

    Instance && \multicolumn{2}{c}{\textsc{SMPT}}
    && \textsc{ITS-Tools} && \textsc{LoLA}
    && \textsc{Tapaal}\\\midrule

         1         && \WIN{0.15} & && \SI{0.78}{}  && \SI{5.01}{} && \SI{0.17}{}\\
         3u        && \SI{1.84}{} &${}^*$&& \SI{0.80}{}  && \WIN{0.01} && \SI{0.16}{}\\
         5pi        && \SI{6.86}{} & && \SI{0.88}{}  && \WIN{0.01} && \SI{0.17}{}\\
         6pi        && \SI{0.21}{} & && \SI{0.88}{}  && \WIN{0.01} && \SI{0.16}{}\\
        7pi        && \WIN{0.15} & && \SI{0.78}{}  && \SI{5.00}{} && \SI{0.16}{}\\
    Crypto    && \SI{0.20}{} &${}^*$&& \SI{4.94}{} &&   \TLE   && \WIN{0.16}\\
 Crypto-10000 && \SI{0.25}{} &${}^*$&& \SI{4.88}{}  && \WIN{0.02} && \SI{0.16}{}\\
   Crypto-50  && \SI{0.22}{} &${}^*$&& \SI{1.04}{}  && \WIN{0.01} && \SI{0.18}{}\\
  Crypto-500  && \SI{0.24}{} &${}^*$&& \SI{4.57}{}  && \WIN{0.01} && \SI{0.17}{}\\
     PGCD-10000    && \WIN{0.14} &${}^*$&& \SI{142.63}{} &&   \TLE  && \SI{96.59}{}\\
      PGCD-50      && \SI{0.10}{} &${}^*$&& \SI{0.87}{}  && \WIN{0.01} && \SI{0.17}{}\\
      PGCD-500     && \SI{0.11}{} &${}^*$&& \SI{1.13}{}  && \WIN{0.08} && \SI{0.30}{}\\
         b         && \WIN{0.09} & && \SI{0.79}{}  && \SI{5.02}{} && \SI{0.16}{}\\
        kw2        && \SI{0.18}{} & && \SI{0.78}{}  && \SI{5.01}{} && \WIN{0.16}\\
        mtx        && \SI{0.60}{} & && \SI{0.86}{}  && \SI{0.00}{} && \WIN{0.16}\\
        nope       && \WIN{0.12} & && \SI{0.78}{} && \SI{5.01}{} && \SI{0.16}{}\\
       nope2       && \WIN{0.10} & && \SI{0.76}{} && \SI{5.01}{} && \SI{0.17}{}\\
       test12      && \SI{0.10}{} & && \SI{0.76}{} && \SI{5.00}{} && \WIN{0.05}\\
       test3       && \WIN{0.15} & && \SI{0.91}{} && \SI{5.02}{} && \SI{0.18}{}\\
       test4       && \SI{0.15}{} & && \SI{0.82}{} && \WIN{0.01} && \SI{0.17}{}\\
         u         && \WIN{0.09} & && \SI{0.77}{} && \SI{5.00}{} && \SI{0.17}{}\\
         w         && \WIN{0.10} & && \SI{0.79}{} && \SI{5.02}{} && \SI{0.16}{}\\
         w1        && \SI{0.10}{} & && \SI{0.75}{} && \SI{5.01}{} && \WIN{0.05}\\
         w2        && \WIN{0.10} & && \SI{0.80}{} && \SI{5.00}{} && \SI{0.17}{}\\
         wb        && \SI{0.29}{} &${}^*$ && \SI{0.80}{} && \SI{5.00}{} && \WIN{0.17}\\
         we        && \WIN{0.16} & && \SI{0.78}{} && \SI{5.01}{} && \SI{0.17}{}\\
         x         && \SI{1.24}{} &${}^*$&& \SI{0.84}{} && \WIN{0.01} && \SI{0.16}{}\\
         z         && \WIN{0.10} & && \SI{1.22}{} && \SI{5.00}{} && \SI{0.30}{}\\
         ze        && \SI{0.71}{} & && \SI{0.88}{} && \SI{5.03}{} && \WIN{0.17}\\
         zz        && \WIN{0.12} &${}^*$&& \SI{1.64}{} && \SI{5.01}{} && \SI{0.25}{}\\
  \end{tabular}
\end{center}
\caption{Computation times with existing benchmarks (time in
  seconds) \label{tab2}}
\end{table}

\begin{figure}
  \centering
  \includegraphics[width=\textwidth]{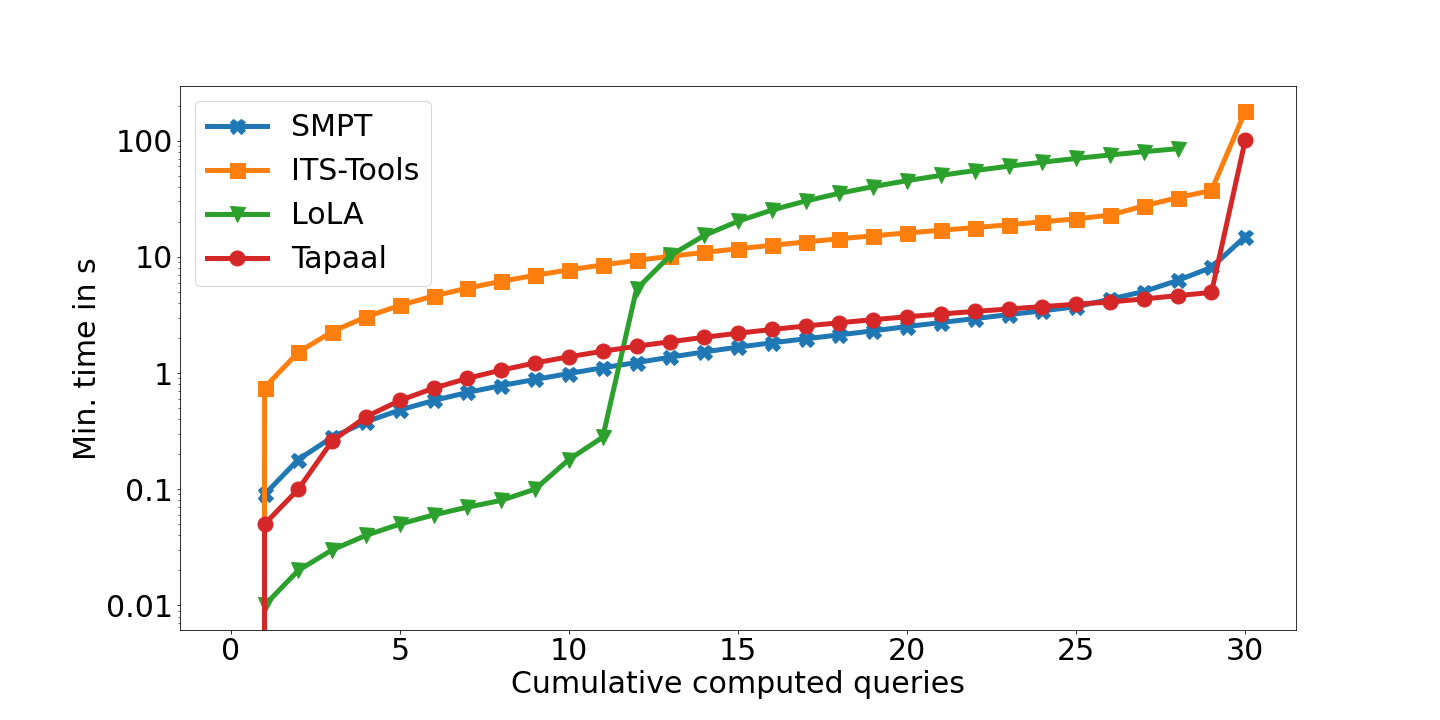}
  \caption{Minimal timeout to compute a given number of queries.}
  \label{fig:cactus}
\end{figure}

Our tool is actually quite mature. In particular, a preliminary
version of SMPT~\cite{pn2021} (without many of the improvements
described in this work) participated in the 2021 edition of the
MCC, where we ranked fourth, out of five competitors,
and achieved a reliability in excess of $99.9\%$.

Even if it was with a previous version of our tool, there are still
lessons to be learned from these results. In particular, it can inform
us on the behavior of \textsc{SMPT} on a very large and diverse
benchmark of bounded nets, with a majority of reachability properties
that are not invariants.

We can compare our results with those of \textsc{LoLA}, that fared
consistently well in the reachability category of the
MCC. \textsc{LoLA} is geared towards model checking of finite state
spaces, but it also implements semi-decision procedures for the
unbounded case. Out of $45\,152$ reachability queries at the MCC in
2021 (one instance of a net with one formula), \textsc{LoLA} was able
to solve $85\%$ of them ($38\,175$ instances) and \textsc{SMPT} only
$52\%$ ($23\,375$ instances); it means approximately $\times 1.6$ more
instances solved using \textsc{LoLA} than using \textsc{SMPT}. Most of
the instances solved with \textsc{SMPT} have also been solved by
\textsc{LoLA}; but still $1\,631$ instances are computed only with our
tool, meaning we potentially increase the number of computed queries
by $4\%$. This is quite an honorable result for \textsc{SMPT},
especially when we consider the fact that we use a single technique,
with only a limited number of optimizations.


\section{Conclusion and Related Works}
\label{sec:conclusion}

One of the most important results in concurrency theory is the
decidability of reachability for Petri nets or, equivalently, for
Vector Addition Systems with States (VASS)~\cite{kosaraju}. Even if
this result is based on a constructive proof, and its ``construction''
streamlined over time~\cite{leroux2009general}, the classical
Kosaraju-Lambert-Mayr-Sacerdote-Tenney approach does not lead to a
workable algorithm. It is in fact a feat that this algorithm has been
implemented at all, see e.g. the tool
\textsc{KReach}~\cite{dixon_kreach_2020}.
While the (very high) complexity of the problem means that no single
algorithm could work efficiently on all inputs, it does not prevent
the existence of methods that work well on some classes of
problems. For example, several algorithms are tailored for the
discovery of counter-examples. We mention the tool
\textsc{FastForward}~\cite{blondin_directed_2021} in our experiments,
that explicitly targets the case of unbounded nets.

We propose a method that works as well on bounded as on unbounded
ones; that behaves well when the invariant is true; and that works with
``genuine'' reachability properties, and not only with
coverability. But there is of course no panacea.  Our approach relies
on the use of linear predicates, which are incrementally strengthened
until we find an invariant based on: the transition relation of the
net; the property we want to prove (it is ``property-directed''); and
constraints on the initial states.
This is in line with a property proved by
Leroux~\cite{leroux2009general}, which states that when a final
configuration is not reachable then ``\emph{there exist checkable
  certificates of non-reachability in the Presburger arithmetic.}''
Our extension of PDR provides a constructive method for computing such
certificates, when it terminates.
For our future works, we would like to study more precisely the
completeness of our approach and/or its limits.

This is not something new. There are many tools that rely on the use
of integer programming techniques to check reachability properties. We
can mention the tool \textsc{Sara}~\cite{wimmel2012applying}, that is
now integrated inside \textsc{LoLA} and can answer reachability
problems on unbounded nets; or libraries like
\textsc{Fast}~\cite{bardin2008fast}, designed for the analysis of
systems manipulating unbounded integer variables. An advantage of our
method is that we proceed in a lazy way. We never explicitly compute
the structural invariants of a net, never switch between a Presburger
formula and its representation as a semilinear set (useful when one
wants to compute the ``Kleene closure'' of a linear constraint), \dots
and instead let a SMT solver work its magic.

We can also mention previous works on adapting PDR/IC3 to Petri
nets. A first implementation of SMPT was presented in~\cite{pn2021},
where we focused on the integration of structural reductions with
PDR. This work did not use our abstraction methods based on hurdles
and saturation, which are new. We can find other related works, such
as~\cite{esparza_smt-based_2014,DBLP:conf/apn/KangBJ21,kloos_incremental_2013}. Nonetheless
they all focus on coverability properties. Coverability is not only a
subclass of the general reachability problem, it has a far simpler
theoretical complexity (EXPSPACE vs NONELEMENTARY). It is also not
expressive enough for checking the absence of deadlocks or for complex
invariants, for instance involving a comparison between the marking of
two places, such as $p < q$. The idea we advocate is that approaches
based on the generalization of markings are not enough. This is why we
believe that abstractions (G2) and (G3) defined in Lemma~\ref{lemma:g}
are noteworthy.

We can also compare our approach with tools oriented to the
verification of bounded Petri nets; since many of them integrate
methods and semi-decision procedures that can work in the unbounded
case. The best performing tools in this category are based on a
portfolio approach and mix different methods. We compared ourselves
with three tools: \textsc{ITS-Tools} \cite{its_tools}, \textsc{Tapaal}
\cite{tapaal} and \textsc{LoLA} \cite{lola}, that have in common to be
the top trio in the Model Checking Contest~\cite{mcc2019}.
(And can therefore accept a common syntax to describe nets and
properties.) Our main contribution in this context, and one of our
most complex results, is to provide a new benchmark of nets and
properties that can be used to evaluate future reachability algorithms
``for expressiveness''.

The methods closest to ours in these portfolios are Bounded
Model Checking and $k$-induction~\cite{sheeran_checking_2000}, which
are also based on the use of SMT solvers. We can mention the case of
\textsc{ITS-Tools}~\cite{thierry-mieg_structural_2020}, that can build
a symbolic over-approximation of the state space, represented as set
of constraints. This approximation is enough when it is included in
the invariant that we check, but inconclusive otherwise.
A subtle and important difference between PDR and these methods is
that PDR needs only $2 n$ variables (the $\vec{p}$ and $\vec{p'}$),
whereas we need $n$ fresh variables at each new iteration of
$k$-induction (so $k n$ variables in total).  This contributes to the
good performances of PDR since the complexity of the SMT problems are
in part relative to the number of variables involved.
Another example of over-approximation is the use of the so-called
``state equation method''~\cite{state_equation}, that can strengthen
the computations of inductive invariants by adding extra constraints,
such as place invariants~\cite{silva1996linear}, siphons and
traps~\cite{esparza_smt-based_2014,Esparza97verificationof}, causality
constraints, etc. We plan to exploit similar constraints in
\textsc{SMPT} to better refine our invariants.

To conclude, our experiments confirm what we already knew: we always
benefit from using a more diverse set of techniques, and are still in
need of new techniques, able to handle new classes of problems. For
instance, we can attribute the good results of \textsc{Tapaal}, in our
experiments, to their implementation of a Trace Abstraction Refinement
(TAR) techniques, guided by
counter-examples~\cite{cassez2017refinement}. The same can be said
with \textsc{LoLA}, that also uses a {CEGAR}-like
method~\cite{wimmel2012applying}. We believe that our approach could
be a useful addition to these techniques.


\subsubsection*{Acknowledgements.} We would like to thank Alex Dixon,
Philip Offtermatt and Yann Thierry-Mieg for their support when
evaluating their respective tools. Their assistance was essential in
improving the quality of our experiments.


\bibliographystyle{splncs04}
\bibliography{bibfile}

\end{document}